\documentclass[11pt]{article}

\usepackage[table]{xcolor}
\definecolor{mycustomgray}{RGB}{238,238,238}  

\usepackage[margin=1.3in]{geometry}

\usepackage[english]{babel}
\usepackage{graphicx}  
\usepackage{amsmath, amssymb, amsthm}  

\usepackage{ragged2e}
\usepackage{dsfont}  

\usepackage{authblk}

\usepackage{tikz}

\date{}

\newcommand{\Z}{\mathds{Z}}

\newcommand{\f}{\flat}
\newcommand{\s}{\sharp}

\newcommand{\dflat}{\tikz[baseline=-1.2mm] \node {\reflectbox{$\flat$}};}
\newcommand{\sflat}{\tikz[baseline=-1.2mm] \node {\reflectbox{$\flat$}$\flat$};}
\newcommand{\dsharp}{\hskip3pt \tikz[baseline=-1.2 mm] {%
  \clip (-2pt,-6pt) rectangle (-.2pt,6pt); \node at (0,0) {$\sharp$}; \hskip3pt}}
\newcommand{\ssharp}{\tikz[baseline=-1.2mm] {%
  \node[inner sep=0.25mm] at (0,0) {$\sharp$};\node at (1.7pt,0.55pt) {$\sharp$};}}

\newcommand{\lan}{\langle}
\newcommand{\ran}{\rangle}

\theoremstyle{plain}
\newtheorem{theorem}{Theorem}[section]
\newtheorem{lemma}[theorem]{Lemma}

\newtheorem{proposition}[theorem]{Proposition}

\theoremstyle{definition}

\newtheorem{remark}[theorem]{Remark}
\newtheorem{example}[theorem]{Example}
\newtheorem{definition}[theorem]{Definition}

\title{Algebraic Structures in Microtonal Music}

\author{Veronica Flynn, Carmen Rovi}

\begin{document}

\maketitle
\begin{abstract}
 We will discuss how certain group theory structures are found in music theory. Western music splits the octave into 12 equal tones called half-steps. We can take this division further and split the octave into 24 equal tones by splitting each half-step in two, called a quarter-step. By assigning each of these 24 notes a number, we can discuss musical actions mathematically. 
 In this paper, we analyze 24-tone microtonal music and explore how musical and harmonic structures in this system can be interpreted in terms of group-theoretic structures. This work extends the study in \cite{dihedralgroups}.
\end{abstract}

\section{Introduction}

When a string is plucked, it vibrates at a certain frequency and produces a certain pitch. If the length of the vibrating portion of the string is then halved by pressing it down at its midpoint, the string vibrates twice as fast, producing a pitch that we perceive as similar to the original one. In musical terms, this new pitch is exactly one octave higher than the original. The relation between a pitch and its equivalent separated by an octave is a fundamental concept in acoustics and a main pillar in all musical systems.  But how about the pitches in between?

In Western music, we divide the octave into 12 different pitches, each separated from the next by semitones or halfsteps; this is an example of a music system. But many other music systems divide the octave in different ways.  Music systems that divide the octave into more than 12 different pitches are referred to as microtonal systems. They use intervals smaller than the half-steps created by the 12 Western pitches. 
For example, some traditional Indian music uses 22 tones in an octave. Another important example is that of Arabic music, which uses a system close to 24 tones in a melodic way.  In \cite{fiftysystem}, music theorist and composer Ben Johnston argues for a 53-tone system because the interval between each note is the smallest our brains can comprehend. 

In this paper, we will consider a 24-quarter-tone system. What makes the 24-tone system particularly interesting is its relationship with the 12-tone Western system. First, the 12 pitches of Western music are still present in the 24-tone system. 
There are different ways to use this microtone system. The two main ways are:
\begin{enumerate}
    \item To use the additional microtones as a colorful melodic effect while the harmony (chords) stay the same as in usual Western 12-note system.
    \item To use the 24 microtone system as a harmonic system in its own right. 
\end{enumerate}

As music theorist Ivan Wyschnedgradsky shows in his ``Manual of Quarter-Tone Harmony", \cite{manual}, we can create the 24-tone system in a similar way to how the 12-tone system was developed. By dividing each half step into two quarter steps, the additional 12 new tones essentially fill in the gaps between the pre-existing 12 tones. As Wyshnedgradsky explains, we can view this 24-tone system as two copies of the 12-tone system, with one shifted up by a quarter step, similar to two interlocked combs.

Many of the ideas and functions in this paper have been adapted from \cite{dihedralgroups} to fit a 24-tone environment. However, there are some unexpected results. Some of these arise from introducing the ``neutral triads", which are chords specific to the 24-tone system. These are introduced in Section \ref{Sec:neutral triads}, and Theorem \ref{Thm: cyclic} is the main result relating to these triads.
The other surprising result is given in Theorem \ref{Thm:two groups}.  Considering functions that are defined analogously to those in \cite{dihedralgroups}, we do not obtain a dihedral group isomorphic to $D_{24}$, which would be the expected result, but rather two groups isomorphic to $D_{12}$. This is an algebraic reflection of the fact that musically, the quarter-tone system can be viewed as two interlocked 12-tone  systems.

\section{Quarter-Tones and Integers Modulo 24}
As mentioned above, in Western music, we use equal-tempered tuning to divide the octave into 12 unique tones listed below:
%
\begin{table}[ht] 
\begin{center}
\rowcolors{1}{mycustomgray}{mycustomgray}
\begin{tabular}{|c|c|c|c|c|c|c|c|c|c|c|c|}
\hline
    C & C$\s$ &D &D$\s$ &E &F &F$\s$ &
    G & G$\s$ & A & A$\s$ &B \\
     & D$\f$ & & E$\f$ & & & G$\f$ & & A$\f$ & 
     &B$\f$ &  \\
\hline
\end{tabular}
\caption{Original tones in the 12-tone system.}
\end{center}
\end{table}

The interval between two consecutive notes is called a \textit{half-step}. 
The notation $\s$ represents raising a pitch by a half-step, while $\f$ represents lowering a pitch by a half-step.
Observe that some pitches have two-letter names; these are called \textit{enharmonic equivalents}.
If  two notes are separated by a whole number of octaves, we say that they are equivalent. Two such pitches would be in the same \textit{pitch class}; thus, in Western music, the octave is divided into 12 unique pitch classes. 
However, we can take this division even further by splitting each half-step in two. The result is a microtonal system that splits the octave into 24 pitch classes:


\begin{table}[ht]
  \begin{center}
    \begin{tabular}{|c|c|c|c|c|c|c|c|c|c|c|c|}
    \hline
        \cellcolor{mycustomgray}C &C$\dsharp$ 
        & \cellcolor{mycustomgray}C$\s$ 
        & C$\ssharp$ 
        & \cellcolor{mycustomgray}D &D$\dsharp$
        & \cellcolor{mycustomgray}D$\s$ 
        &D$\ssharp$ 
        &\cellcolor{mycustomgray} E 
        &E $\dsharp$ 
        &\cellcolor{mycustomgray}F 
        &F$\dsharp$ \\
\cellcolor{mycustomgray}
        &D$\sflat$ 
        &\cellcolor{mycustomgray}D $\f$ 
        & D$\dflat$ 
        & \cellcolor{mycustomgray}
        & E$\sflat$ 
        & \cellcolor{mycustomgray}E$\f$ 
        & E$\dflat$ 
        & \cellcolor{mycustomgray}
        &F$\dflat$ 
        & \cellcolor{mycustomgray}
        & G$\sflat$\\
    \hline
    \end{tabular} $\rightarrow$
    
  \begin{tabular}{|c|c|c|c|c|c|c|c|c|c|c|c|}
    \hline
         \cellcolor{mycustomgray} F$\s$ 
         & F$\ssharp$ 
         & \cellcolor{mycustomgray} G 
         & G$\dsharp$ 
         & \cellcolor{mycustomgray}G$\s$ 
         & G$\ssharp$ 
         & \cellcolor{mycustomgray}A 
         & A$\dsharp$ 
         & \cellcolor{mycustomgray}A$\s$ 
         & A$\ssharp$ 
         & \cellcolor{mycustomgray}B 
         & B\dsharp\\
         \cellcolor{mycustomgray} G$\f$
         & G$\dflat$ 
         & \cellcolor{mycustomgray}
         &A$\sflat$ 
         & \cellcolor{mycustomgray} A$\f$ 
         & A$\dflat$ 
         & \cellcolor{mycustomgray}
         &B$\sflat$ 
         & \cellcolor{mycustomgray} B$\f$
         & B$\dflat$ 
         & \cellcolor{mycustomgray}
         & C$\dflat$  \\
    \hline
    \end{tabular}
    \caption{The \textit{original tones} are shown here in grey, the \textit{new tones} appear in white.}
\end{center}
\end{table}

Here, the interval between two consecutive notes is called a \textit{quarter-step} and is half the interval of a half-step. 
In summary, the various symbols to alter a given pitch by quarter-steps are as follows:
\begin{itemize}
    \item Raising a note:
    \begin{enumerate}
        \item by one quarter-step: $\dsharp$
        \item by two quarter-steps (or equivalently a half-step): $\sharp$
        \item by three quarter-steps: $\ssharp$
    \end{enumerate}
    \item Lowering a note: 
      \begin{enumerate}
        \item by one quarter-step: $\dflat$
        \item by two quarter-steps (or equivalently a half-step): $\flat$
        \item by three quarter-steps: $\sflat$
    \end{enumerate}
\end{itemize}
\begin{example}
The note $G \ssharp$ would be read $G$ three-quarters sharp, and $A \dflat$ would read $A$ quarter-flat. These notes are enharmonic as they will sound the same.
\end{example}
Moving forward, we will refer to the notes from the Western 12-tone system as the \textit{original tones}, and the additional 12 pitch classes will be the \textit{new tones}. 

The 24-tone system is an example of a microtonal system. What makes this system particularly interesting is its relation to the original 12-tone system. The original 12 tones are still present in this system, so we can discuss musical structures like harmony as we normally would in traditional 12-tone music. Moreover, it is discussed in \cite{manual} that from a musical perspective, one can view the 24-tone system as two copies of the 12 pitch classes, where one is shifted up by a quarter step, similar to two interlocked combs. 

Music theorists have found it useful to translate each of the original 12 pitch classes to integers modulo 12, with 0 being $C$. 
Similarly, we can translate each of the 24 pitch classes to integers \textit{modulo 24}, with 0 as $C$. 
Doing this allows us to create a musical clock, shown in Figure \ref{fig: clock}.

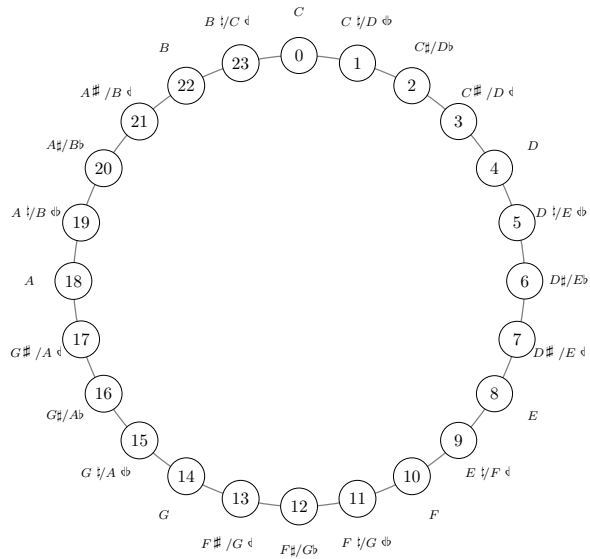
\begin{figure}[ht]
\centering
\scalebox{0.6}{
\begin{tikzpicture}
  \def\radius{5}
  \def\nodesize{0.8}
  \def\gapangle{8}
  \def\labeloffset{0.98} 

  \newcommand{\notelabel}[1]{%
    \ifcase#1 $C$%
    \or $C\dsharp$/$D \sflat$%
    \or $C\sharp$/$D \flat$%
    \or $C\ssharp$/$D \dflat$%
    \or $D$%
    \or \hspace{2pt} $D\dsharp$/$E \sflat$%
    \or $D\sharp$/$E \flat$%
    \or \hspace{2pt}$D\ssharp$/$E \dflat$%
    \or $E$%
    \or $E\dsharp$/$F \dflat$%
    \or $F$%
    \or $F\dsharp$/$G \sflat$%
    \or $F\sharp$/$G \flat$%
    \or $F\ssharp$/$G \dflat$%
    \or $G$%
    \or $G\dsharp$/$A \sflat$%
    \or $G\sharp$/$A \flat$%
    \or $G\ssharp$/$A \dflat$%
    \or $A$%
    \or $A\dsharp$/$B \sflat$%
    \or $A\sharp$/$B \flat$%
    \or $A\ssharp$/$B \dflat$%
    \or $B$%
    \or $B\dsharp$/$C \dflat$%
    \fi
  }

  \foreach \i in {0,...,23} {
    \pgfmathsetmacro{\startangle}{90 - 360*\i/24 - \gapangle/2}
    \pgfmathsetmacro{\endangle}{90 - 360*(\i+1)/24 + \gapangle/2}
    \draw[thick, gray]
      ({\radius*cos(\startangle)}, {\radius*sin(\startangle)})
        arc[start angle=\startangle, end angle=\endangle, radius=\radius];
  }

  \foreach \i in {0,...,23} {
    \pgfmathsetmacro{\angle}{90 - 360*\i/24}
    \pgfmathsetmacro{\x}{\radius*cos(\angle)}
    \pgfmathsetmacro{\y}{\radius*sin(\angle)}

    \node[circle, draw, fill=white, minimum size=\nodesize cm]
      at (\x, \y) {\i};

    \pgfmathsetmacro{\lx}{(\radius + \labeloffset)*cos(\angle)}
    \pgfmathsetmacro{\ly}{(\radius + \labeloffset)*sin(\angle)}
    \node[font=\scriptsize, align=center] at (\lx, \ly) {\notelabel{\i}};
  }

\end{tikzpicture}
}
 \caption{24 Pitches on the Musical Clock.}
    \label{fig: clock}
\end{figure}

Moving one notch around the clock correlates to adding by one and moving up by a quarter-step.
We can use this clock to easily determine the interval (or number of quarter steps) between two notes. 
\begin{example}
   If we want to know the number of quarter-steps between $B$ and $D\ssharp$, we can see that on the clock, the pitch class $B = 22$ and $D\ssharp = 7$. The distance between these two tones is
   \begin{equation}
       (7 - 22) = -15 \equiv 9 \pmod{24}.
   \end{equation}
\end{example}
The translation from the original 12 tones to arithmetic modulo 12 has been previously discussed in several articles, including  \cite{dihedralgroups}. This translation allows us to analyze musical actions mathematically. In the next three sections, we will use a similar approach, but in our case, using modulo 24 arithmetic to analyze the 24-microtone system.

\section{Transposition and Inversion}
Two of the musical actions we will discuss are transposition and inversion. Musically, \textit{transposition} is shifting a note up or down by a number of quarter tones. In many pop songs, it's easy to recognise the same melody sung higher or lower. This is a transposition of the melody. More generally, transposition is used by singers adapting a composition to their personal voice range. 
Transposition is also used structurally in composition.  For example, in a fugue, we first hear the subject (or main theme) and immediately afterwards we can listen to the same theme at a distance of a fifth (See Figure \ref{fig:fugue}). So with transposition, the relationship between the notes of the original melody remains the same, even if the pitch of each note is changed.

\begin{figure}[htbp]
    \centering
    \includegraphics[width=0.7\linewidth, trim=0 0 0 8.5, clip]{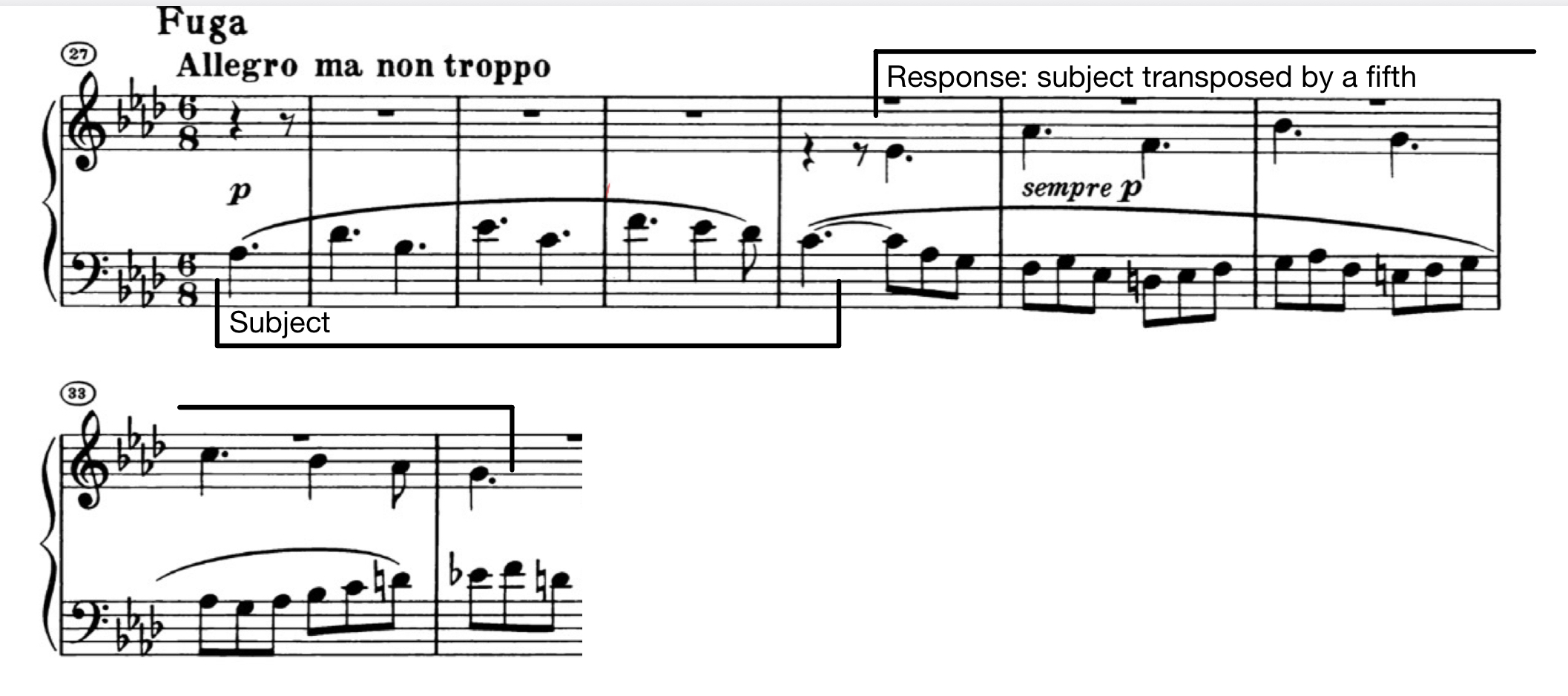}
    \caption{Subject of a fugue (op.\ 110, Beethoven) and the response transposed a fifth higher.}
    \label{fig:fugue}
\end{figure}

On the other hand, \textit{inversion} is like flipping a melody on its head. The melody has the same pattern of intervals but moves in the opposite direction. For example, if a melody starts on a $C$ and moves up 5 quarter-steps, the inversion would be starting on $C$ and moving down 5 quarter-steps. 
We can define these two musical actions as functions using the musical clock.

\begin{definition} \label{Def:transposition and inversion formulas}
Mathematically, the Transposition and Inversion functions are defined as follows:
\begin{enumerate}
    \item The Transposition function $T_n$ is given by 
\begin{center}   
   $T_n: \Z_{24} \longrightarrow \Z_{24}$  \\
    $T_n(x) := x + n \text{ mod } 24$,
    \end{center}
where $x$ is the starting pitch and $n$ is the number of quarter-steps we transpose by.
 \item The Inversion function $I_n$ is given by
\begin{center}
    $I_n : \Z_{24} \longrightarrow \Z_{24}$ \\
    $I_n(x) := -x + n \text{ mod } 24$,
\end{center}
where $n$ is the pitch we are inverting over and $x$ is the pitch that is being inverted. Geometrically, there is an axis of reflection in the musical clock spanned by the line through $\frac{n}{2}$ and $\frac{n}{2}+12$, and $x$ is reflected on that axis.
 \end{enumerate}
\end{definition}

\begin{remark}
    Note that in \cite{dihedralgroups}, the authors also use functions labeled $T_n$ and $I_n$.  Our functions are close analogs of these functions, but they are not identical to them. The context in \cite{dihedralgroups} is the analysis of a 12-tone system, so for them, $T_n$ and $I_n$ are taken modulo 12 rather than modulo 24. 
\end{remark}

The transposition and inversion functions $T_n$ and $I_n$ can be represented on the musical clock. 
\begin{example}
The transposition $T_1$ correlates to a rotation of the clock by $\frac{1}{24}$, and the inversion $I_0$ correlates to a flip of the clock along the 0-12 axis.
\end{example}
As shown in \cite{dihedralgroups}, in the original 12-tone system, $T_n$ and $I_n$ generate the dihedral group of a 12-gon. Similarly, in the 24-tone system, we can formulate an analog result.
\begin{proposition}
 The modulo 24 transposition and inversion functions $T_1$ and $I_0$ generate the dihedral group of symmetries of a 24-gon. 
\end{proposition}
\begin{proof}   

   One can readily check these relations: 
\begin{align}
    (T_1)^n &= T_{n \bmod 24} \notag \\
    T_m \circ T_n &= T_{m + n \bmod 24} \notag \\
    T_m \circ I_n &= I_{m + n \bmod 24} \notag \\
    I_m \circ T_n &= I_{m - n \bmod 24} \notag \\
    I_m \circ I_n &= T_{m -n \bmod 24}, \notag 
\end{align}
We can see that the set generated by $T_1$ and $I_0$, together with the relations above, forms a group: any combination of transpositions and inversions will yield another transposition or inversion, so there is closure. 
There is an identity element $1=T_0$. Each element has an inverse, and composition is associative.

Further, 
from the relations above we observe that $(T_1)^{24} = 1 = (I_0)^2$, and $I_m \circ T_n \circ I_m = (T_n)^{-1}$. So we can see that this group is isomorphic to the dihedral group with 48 elements, $D_{24}$,
$$ T I = \{\langle T_1, I_0 \rangle\ ~\vert~ (T_1)^{24} = 1=(I_0)^2, I_m \circ T_n \circ I_m = (T_n)^{-1}  \} \cong D_{24}.  \qedhere $$
\end{proof}
\raggedbottom

\section{Triad chords}
We can use the musical clock to analyze even more complex musical objects: chords.
In music, a chord is multiple notes played simultaneously. Because these notes are played at the same time, we can view chords as unordered sets where the elements of the sets are numbers that correlate to the pitches in the chord. We will specifically be looking at three types of three-note chords: \textit{major triads}, \textit{minor triads}, and \textit{neutral triads}.  Major and minor triads are the usual basic triads in the 12-note system, while the neutral triad is exclusive of the 24-note system (see \cite{manual}).
\begin{definition}\label{def:triads} The three different types of triads in a 24-tone system are:
\begin{enumerate}
    \item A \textbf{major triad} consists of a root note $x$, a second note 8 quarter-tones above the root, and a third note 14 quarter-tones above the root. Thus, a major triad takes the form $\langle x, x + 8, x + 14 \rangle$.
    \item A \textbf{minor triad} consists of a root note, a second note 6 quarter-tones above the root, and a third note 14 quarter-tones above the root. Thus, a minor triad takes the form $\langle x, x +6, x + 14 \rangle$.
    \item A \textbf{neutral triad} takes the form $\langle x, x + 7, x + 14\rangle$. This type of chord is not present in the 12-tone system. 
\end{enumerate}
    \end{definition}
    \begin{example}
A $C$ major triad has $C$ as a root note. So $x = 0$, and we denote the $C$ major triad as $\langle 0, 8, 14\rangle = \langle C, E, G\rangle$. A $c$ minor triad is denoted by $\langle 0, 6, 14\rangle$.

An $f$ minor triad has $F$ as its root and is denoted $\langle 10, 16, 0\rangle = \langle F, A, C \rangle$. 
    \end{example}
    
Triads can be represented as triangles on the musical clock. For example, Figure \ref{fig:C Major Triad}  shows a $C$ major triad as a triangle on the musical clock.

Because the elements of triads are pitch classes in $\Z_{24}$, $T_n$ and $I_n$ are applied to them component-wise:
$$
T_n(\langle y_1, y_2, y_3 \rangle) = \langle T_n(y_1), T_n(y_2), T_n(y_3) \rangle, 
$$
$$
I_n(\langle y_1, y_2, y_3 \rangle) = \langle I_n(y_1), I_n(y_2), I_n(y_3) \rangle.
$$

\begin{example}
To find the inversion of a $C$ major triad, we apply $I_0$ to $C$ as follows: $I_0(C) = I_0\langle 0, 8, 14 \rangle = \langle I_0 (0), I_0(8), I_0(14) \rangle = \langle 0, 16, 10 \rangle = f$.   Figure \ref{fig: C Major Triad inverted over 0} shows this inversion on the musical clock. 
\end{example}

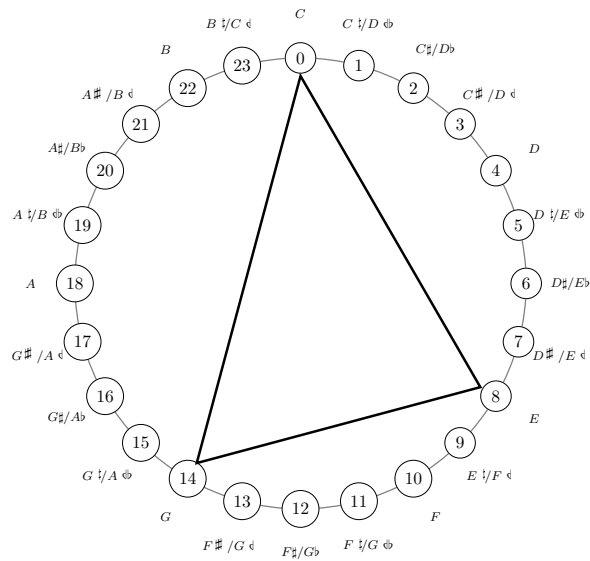
\begin{figure}[htbp]
\centering
\scalebox{0.6}{
\begin{tikzpicture}
  \def\radius{5}
  \def\nodesize{0.6}
  \def\gapangle{8}
  \def\labeloffset{0.98}
  \def\innertriangleoffset{0.4} 

  \newcommand{\notelabel}[1]{%
    \ifcase#1 $C$%
    \or $C\dsharp$/$D \sflat$%
    \or $C\sharp$/$D \flat$%
    \or $C\ssharp$/$D \dflat$%
    \or $D$%
    \or $D\dsharp$/$E \sflat$%
    \or $D\sharp$/$E \flat$%
    \or $D\ssharp$/$E \dflat$%
    \or $E$%
    \or $E\dsharp$/$F \dflat$%
    \or $F$%
    \or $F\dsharp$/$G \sflat$%
    \or $F\sharp$/$G \flat$%
    \or $F\ssharp$/$G \dflat$%
    \or $G$%
    \or $G\dsharp$/$A \sflat$%
    \or $G\sharp$/$A \flat$%
    \or $G\ssharp$/$A \dflat$%
    \or $A$%
    \or $A\dsharp$/$B \sflat$%
    \or $A\sharp$/$B \flat$%
    \or $A\ssharp$/$B \dflat$%
    \or $B$%
    \or $B\dsharp$/$C \dflat$%
    \fi
  }

  \foreach \i in {0,...,23} {
    \pgfmathsetmacro{\startangle}{90 - 360*\i/24 - \gapangle/2}
    \pgfmathsetmacro{\endangle}{90 - 360*(\i+1)/24 + \gapangle/2}
    \draw[thick, gray]
      ({\radius*cos(\startangle)}, {\radius*sin(\startangle)})
        arc[start angle=\startangle, end angle=\endangle, radius=\radius];
  }

  \foreach \i in {0,...,23} {
    \pgfmathsetmacro{\angle}{90 - 360*\i/24}
    \pgfmathsetmacro{\x}{\radius*cos(\angle)}
    \pgfmathsetmacro{\y}{\radius*sin(\angle)}

    \node[circle, draw, fill=white, minimum size=\nodesize cm]
      at (\x, \y) {\i};

    \pgfmathsetmacro{\lx}{(\radius + \labeloffset)*cos(\angle)}
    \pgfmathsetmacro{\ly}{(\radius + \labeloffset)*sin(\angle)}
    \node[font=\scriptsize, align=center] at (\lx, \ly) {\notelabel{\i}};
  }

  \pgfmathsetmacro{\aone}{90 - 360*0/24}
  \pgfmathsetmacro{\atwo}{90 - 360*8/24}
  \pgfmathsetmacro{\athree}{90 - 360*14/24}

  \pgfmathsetmacro{\rtri}{\radius - \innertriangleoffset}

  \coordinate (A) at ({\rtri*cos(\aone)}, {\rtri*sin(\aone)});
  \coordinate (B) at ({\rtri*cos(\atwo)}, {\rtri*sin(\atwo)});
  \coordinate (C) at ({\rtri*cos(\athree)}, {\rtri*sin(\athree)});

  \draw[ultra thick, black] (A) -- (B) -- (C) -- cycle;

\end{tikzpicture}
}
\caption{$C$ Major triad.}
    \label{fig:C Major Triad}
\end{figure}

\begin{figure}[htbp]
\centering
\scalebox{0.6}{
\begin{tikzpicture}
  \def\radius{5}
  \def\nodesize{0.6}
  \def\gapangle{8}
  \def\labeloffset{0.98}
  \def\innertriangleoffset{0.4}
  \pgfmathsetmacro{\rtri}{\radius - \innertriangleoffset}

  \newcommand{\notelabel}[1]{%
    \ifcase#1 $C$%
    \or $C\dsharp$/$D \sflat$%
    \or $C\sharp$/$D \flat$%
    \or $C\ssharp$/$D \dflat$%
    \or $D$%
    \or $D\dsharp$/$E \sflat$%
    \or $D\sharp$/$E \flat$%
    \or $D\ssharp$/$E \dflat$%
    \or $E$%
    \or $E\dsharp$/$F \dflat$%
    \or $F$%
    \or $F\dsharp$/$G \sflat$%
    \or $F\sharp$/$G \flat$%
    \or $F\ssharp$/$G \dflat$%
    \or $G$%
    \or $G\dsharp$/$A \sflat$%
    \or $G\sharp$/$A \flat$%
    \or $G\ssharp$/$A \dflat$%
    \or $A$%
    \or $A\dsharp$/$B \sflat$%
    \or $A\sharp$/$B \flat$%
    \or $A\ssharp$/$B \dflat$%
    \or $B$%
    \or $B\dsharp$/$C \dflat$%
    \fi
  }

  \foreach \i in {0,...,23} {
    \pgfmathsetmacro{\startangle}{90 - 360*\i/24 - \gapangle/2}
    \pgfmathsetmacro{\endangle}{90 - 360*(\i+1)/24 + \gapangle/2}
    \draw[thick, gray]
      ({\radius*cos(\startangle)}, {\radius*sin(\startangle)})
        arc[start angle=\startangle, end angle=\endangle, radius=\radius];
  }

  \pgfmathsetmacro{\aZero}{90 - 360*0/24}
  \pgfmathsetmacro{\aTwelve}{90 - 360*12/24}
  \def\lineinner{4.3}  
  \def\lineouter{5.7}  
  \coordinate (LStart) at ({\lineinner*cos(\aZero)}, {\lineinner*sin(\aZero)});
  \coordinate (LEnd) at ({\lineouter*cos(\aTwelve)}, {\lineouter*sin(\aTwelve)});
  \draw[dotted, thick] (LStart) -- (LEnd);

  \foreach \i in {0,...,23} {
    \pgfmathsetmacro{\angle}{90 - 360*\i/24}
    \pgfmathsetmacro{\x}{\radius*cos(\angle)}
    \pgfmathsetmacro{\y}{\radius*sin(\angle)}

    \node[circle, draw, fill=white, minimum size=\nodesize cm]
      at (\x, \y) {\i};

    \pgfmathsetmacro{\lx}{(\radius + \labeloffset)*cos(\angle)}
    \pgfmathsetmacro{\ly}{(\radius + \labeloffset)*sin(\angle)}
    \node[font=\scriptsize, align=center] at (\lx, \ly) {\notelabel{\i}};
  }

  \pgfmathsetmacro{\aEight}{90 - 360*8/24}
  \pgfmathsetmacro{\aFourteen}{90 - 360*14/24}
  \coordinate (T1A) at ({\rtri*cos(\aZero)}, {\rtri*sin(\aZero)});
  \coordinate (T1B) at ({\rtri*cos(\aEight)}, {\rtri*sin(\aEight)});
  \coordinate (T1C) at ({\rtri*cos(\aFourteen)}, {\rtri*sin(\aFourteen)});
  \draw[ultra thick, black] (T1A) -- (T1B) -- (T1C) -- cycle;

  \pgfmathsetmacro{\aTen}{90 - 360*10/24}
  \pgfmathsetmacro{\aSixteen}{90 - 360*16/24}
  \coordinate (T2A) at ({\rtri*cos(\aZero)}, {\rtri*sin(\aZero)});
  \coordinate (T2B) at ({\rtri*cos(\aTen)}, {\rtri*sin(\aTen)});
  \coordinate (T2C) at ({\rtri*cos(\aSixteen)}, {\rtri*sin(\aSixteen)});
  \draw[ultra thick, blue] (T2A) -- (T2B) -- (T2C) -- cycle;

\end{tikzpicture}
}
   \caption{$I_0(C) = f$}
 \label{fig: C Major Triad inverted over 0}
\end{figure}

\subsection{Consonant (Major and Minor) Triads}
Major and minor triads are grouped together as \textit{consonant triads} because of their smooth, non-grating sound. We define $S$ as the set of all consonant triads. Each of the 24 pitches acts as the root of both a major and minor triad, so $|S| = 48$. Figure \ref{tab: consonant triads} shows 14 of the consonant triads. A capital letter represents a major triad, and a lower case represents a minor triad. 
\begin{figure}[ht!]
    \centering
    \begin{tabular}{|c|c|}
    \hline
    Major & Minor \\
    \hline
    $C = \langle0, 8, 14\rangle$ & $\langle0, 16, 10\rangle = f$ \\    
    $C\dsharp = D\sflat = \lan1, 9, 15\ran$ & $\lan1, 17, 11\ran = f\dsharp = g\sflat$  \\
    $C \s = D\f   = \lan2, 10, 16\ran$  & $\lan2, 18, 12\ran f\s = g\f$  \\
    $C\ssharp = D\dflat = \lan3, 11, 17\ran$ & $\lan3, 19, 13\ran = f\ssharp = g\dflat$ \\
    $D = \lan4, 12, 18\ran$ & $\lan4, 20, 14\ran = g$ \\
    $D\dsharp = E\sflat = \lan5, 13, 19\ran$ &  $ \lan5, 21, 15\ran = g\dsharp = a\sflat$ \\
    $D\s = E\f = \lan6, 14, 20\ran$ & $ \lan6, 22, 16\ran = g\s = a
    \f$ \\
    \vdots & \vdots \\
    \hline    
    \end{tabular}
    \caption{Consonant Triads}
    \label{tab: consonant triads}
\end{figure}
\begin{proposition}\label{prop:transitive on S}
    The action of the $TI$ group on the set $S$ of consonant triads is simply transitive.
\end{proposition}
\begin{proof}
 As in \cite{dihedralgroups}, we can use the table in Figure \ref{tab: consonant triads} to see that the action of the $TI$ group is \textit{simply transitive} on the set $S$. This means that given any two consonant triads $X$ and $Y$, there is a unique $h \in TI \text{ Group}$ such that $hX = Y$.

Observe that an application of $T_n$ to an entry yields the triad $n$ rows below it. 
That is, counting the first row as 0, the $n$-th entry in the first column is 
\begin{align}
    T_n(C) = T_n\lan0, 8, 14\ran &= \lan T_n(0), T_n(8), T_n(14)\ran = \lan n, n+8, n+14 \ran  \notag
\end{align}
or the translation of $C$ up by $n$ quarter-steps. Further, applying $I_0$ to a triad will yield the triad in the same row but in the opposite column. As shown above $I_0(C) = f$. Thus the $n$-th entry of the second column is 
\begin{align}
    I_n\lan0, 8, 14\ran &= \lan I_n(0), I_n(8), I_n(14) \ran.  \notag 
\end{align}
Hence, given two consonant triads $X$ and $Y$, there are elements $h_1, h_2$ in the $TI$ group such that $h_1C = X$ and $h_2C = Y$. Then with $h = h_1h_2^{-1}$, $hX = Y$. One can readily check that this $h$ is unique. Therefore, the action of the $TI$ group is simply transitive. 
\end{proof}

\subsection{Neutral Triads}\label{Sec:neutral triads}
The introduction of 12 new tones creates numerous new intervals and chords that do not exist in the original 12-tone system. One such chord Wyshnedgradsky denotes \textit{neutral triad} and takes the form $\langle x, x + 7, x + 14 \rangle$, introduced in Definition \ref{def:triads}.  
\begin{figure}[ht!]
    \centering
    \begin{tabular}{|c|}
    \hline
    Neutral Triads\\ 
    \hline 
         $C = \langle0, 7, 14\rangle$   \\
         $C\dsharp = D \sflat = \langle 1, 8, 15 \rangle$ \\
         $C\sharp = D\flat = \langle 2, 9, 16 \rangle$ \\
         $C\ssharp = D\dflat = \langle 3, 10, 17 \rangle $ \\
         $D = \langle 4, 11, 18 \rangle $\\
         $D\dsharp = E\sflat = \langle 5, 12, 19 \rangle$ \\
         $D \sharp = E\flat = \langle 6, 13, 20 \rangle $ \\
         $D\ssharp = E\dflat = \langle 7, 14, 21 \rangle $ \\
         $E = \langle 8, 15, 22 \rangle$ \\
         $E \dsharp = F\dflat = \langle 9, 16, 23 \rangle$ \\
         $F = \langle 10, 17, 0 \rangle $\\
         $F\dsharp = G\sflat = \langle 11, 18, 1 \rangle $\\
         \vdots \\
    \hline
    \end{tabular}
    \caption{Neutral Triads }
    \label{Table neutral traids }
\end{figure}
Now define the set $N$ as the set of all neutral triads. Note that $\vert N \vert = 24$ as we can construct a neutral triad with root note any of the 24 notes.   The first 12 neutral triads are listed in Figure 8.  
We can represent neutral triads, their inversions, and transpositions on the musical clock. Figure \ref{fig:neutral C triad} shows a $C$ neutral triad on the clock, and Figure \ref{fig:inverting C neutral} shows the inversion of the $C$ neutral triad. Note that applying any element of the $TI$ group to a neutral triad produces again a neutral triad. In particular, we note the following Lemma.

\begin{lemma}
Any inversion acting on a neutral triad can be expressed as a transposition as
$$I_n \langle x, x+7, x+14 \rangle = T_{-2x+n+10} \langle x, x+7, x+14 \rangle.$$
\end{lemma}

\begin{proof}
    The proof follows from an application of Definition \ref{Def:transposition and inversion formulas}.
    \begin{align*}
I_n \langle x, x+7, x+14 \rangle &= \langle -x+n, -x-7+n, -x-14+n \rangle \\
&= \langle -x-14 +n, -x-7+n, -x+n \rangle \\
&= \langle -x+n+10, -x+n+17, -x+n \rangle \\
&= T_{-2x+n+10} \langle x, x+7, x+14 \rangle. \qedhere 
    \end{align*} 
\end{proof}

\begin{remark}
Note that while the action of $TI$ on the set of neutral triads $N$ is transitive, it is not simply transitive like the action of $TI$ on the set of consonant triads $S$.
\end{remark}


\begin{figure}[htbp]
\centering
\scalebox{0.6}{
\begin{tikzpicture}
  \def\radius{5}
  \def\nodesize{0.6}
  \def\gapangle{8}
  \def\labeloffset{0.98}
  \def\innertriangleoffset{0.4}
  \pgfmathsetmacro{\rtri}{\radius - \innertriangleoffset}

  \newcommand{\notelabel}[1]{%
    \ifcase#1 $C$%
    \or $C\dsharp$/$D \sflat$%
    \or $C\sharp$/$D \flat$%
    \or $C\ssharp$/$D \dflat$%
    \or $D$%
    \or $D\dsharp$/$E \sflat$%
    \or $D\sharp$/$E \flat$%
    \or $D\ssharp$/$E \dflat$%
    \or $E$%
    \or $E\dsharp$/$F \dflat$%
    \or $F$%
    \or $F\dsharp$/$G \sflat$%
    \or $F\sharp$/$G \flat$%
    \or $F\ssharp$/$G \dflat$%
    \or $G$%
    \or $G\dsharp$/$A \sflat$%
    \or $G\sharp$/$A \flat$%
    \or $G\ssharp$/$A \dflat$%
    \or $A$%
    \or $A\dsharp$/$B \sflat$%
    \or $A\sharp$/$B \flat$%
    \or $A\ssharp$/$B \dflat$%
    \or $B$%
    \or $B\dsharp$/$C \dflat$%
    \fi
  }

  \foreach \i in {0,...,23} {
    \pgfmathsetmacro{\startangle}{90 - 360*\i/24 - \gapangle/2}
    \pgfmathsetmacro{\endangle}{90 - 360*(\i+1)/24 + \gapangle/2}
    \draw[thick, gray]
      ({\radius*cos(\startangle)}, {\radius*sin(\startangle)})
        arc[start angle=\startangle, end angle=\endangle, radius=\radius];
  }

  \foreach \i in {0,...,23} {
    \pgfmathsetmacro{\angle}{90 - 360*\i/24}
    \pgfmathsetmacro{\x}{\radius*cos(\angle)}
    \pgfmathsetmacro{\y}{\radius*sin(\angle)}

    \node[circle, draw, fill=white, minimum size=\nodesize cm]
      at (\x, \y) {\i};

    \pgfmathsetmacro{\lx}{(\radius + \labeloffset)*cos(\angle)}
    \pgfmathsetmacro{\ly}{(\radius + \labeloffset)*sin(\angle)}
    \node[font=\scriptsize, align=center] at (\lx, \ly) {\notelabel{\i}};
  }

  \pgfmathsetmacro{\aZero}{90 - 360*0/24}
  \pgfmathsetmacro{\aSeven}{90 - 360*7/24}
  \pgfmathsetmacro{\aFourteen}{90 - 360*14/24}

  \def\innertriangleoffset{0.4}
  \pgfmathsetmacro{\rtri}{\radius - \innertriangleoffset}

  \coordinate (T0) at ({\rtri*cos(\aZero)}, {\rtri*sin(\aZero)});
  \coordinate (T7) at ({\rtri*cos(\aSeven)}, {\rtri*sin(\aSeven)});
  \coordinate (T14) at ({\rtri*cos(\aFourteen)}, {\rtri*sin(\aFourteen)});
  \draw[thick, black] (T0) -- (T7) -- (T14) -- cycle;

\end{tikzpicture}
}
 \caption{$C_n$ chord, i.e. a neutral chord with root note $C$. }
 \label{fig:neutral C triad}
\end{figure}
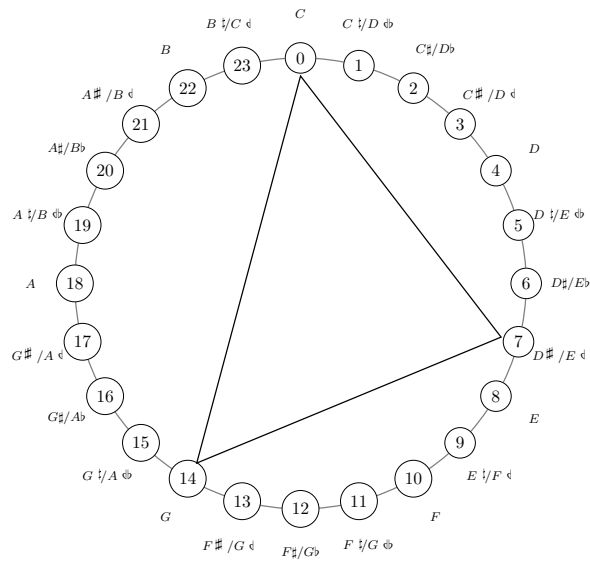


\begin{figure}[htbp]
\centering
\scalebox{0.6}{
\begin{tikzpicture}
  \def\radius{5}
  \def\nodesize{0.6}
  \def\gapangle{8}
  \def\labeloffset{0.98}
  \def\innertriangleoffset{0.4}
  \pgfmathsetmacro{\rtri}{\radius - \innertriangleoffset}

  \newcommand{\notelabel}[1]{%
    \ifcase#1 $C$%
    \or $C\dsharp$/$D \sflat$%
    \or $C\sharp$/$D \flat$%
    \or $C\ssharp$/$D \dflat$%
    \or $D$%
    \or $D\dsharp$/$E \sflat$%
    \or $D\sharp$/$E \flat$%
    \or $D\ssharp$/$E \dflat$%
    \or $E$%
    \or $E\dsharp$/$F \dflat$%
    \or $F$%
    \or $F\dsharp$/$G \sflat$%
    \or $F\sharp$/$G \flat$%
    \or $F\ssharp$/$G \dflat$%
    \or $G$%
    \or $G\dsharp$/$A \sflat$%
    \or $G\sharp$/$A \flat$%
    \or $G\ssharp$/$A \dflat$%
    \or $A$%
    \or $A\dsharp$/$B \sflat$%
    \or $A\sharp$/$B \flat$%
    \or $A\ssharp$/$B \dflat$%
    \or $B$%
    \or $B\dsharp$/$C \dflat$%
    \fi
  }

  \foreach \i in {0,...,23} {
    \pgfmathsetmacro{\startangle}{90 - 360*\i/24 - \gapangle/2}
    \pgfmathsetmacro{\endangle}{90 - 360*(\i+1)/24 + \gapangle/2}
    \draw[thick, gray]
      ({\radius*cos(\startangle)}, {\radius*sin(\startangle)})
        arc[start angle=\startangle, end angle=\endangle, radius=\radius];
  }

  \pgfmathsetmacro{\aZero}{90 - 360*0/24}
  \pgfmathsetmacro{\aTwelve}{90 - 360*12/24}
  \def\lineinner{4.3}  
  \def\lineouter{5.7}  
  \coordinate (LStart) at ({\lineinner*cos(\aZero)}, {\lineinner*sin(\aZero)});
  \coordinate (LEnd) at ({\lineouter*cos(\aTwelve)}, {\lineouter*sin(\aTwelve)});
  \draw[dotted, thick] (LStart) -- (LEnd);

  \foreach \i in {0,...,23} {
    \pgfmathsetmacro{\angle}{90 - 360*\i/24}
    \pgfmathsetmacro{\x}{\radius*cos(\angle)}
    \pgfmathsetmacro{\y}{\radius*sin(\angle)}

    \node[circle, draw, fill=white, minimum size=\nodesize cm]
      at (\x, \y) {\i};

    \pgfmathsetmacro{\lx}{(\radius + \labeloffset)*cos(\angle)}
    \pgfmathsetmacro{\ly}{(\radius + \labeloffset)*sin(\angle)}
    \node[font=\scriptsize, align=center] at (\lx, \ly) {\notelabel{\i}};
  }

  \pgfmathsetmacro{\aSeven}{90 - 360*7/24}
  \pgfmathsetmacro{\aFourteen}{90 - 360*14/24}
  \coordinate (T1A) at ({\rtri*cos(\aZero)}, {\rtri*sin(\aZero)});
  \coordinate (T1B) at ({\rtri*cos(\aSeven)}, {\rtri*sin(\aSeven)});
  \coordinate (T1C) at ({\rtri*cos(\aFourteen)}, {\rtri*sin(\aFourteen)});
  \draw[ultra thick, black] (T1A) -- (T1B) -- (T1C) -- cycle;

  \pgfmathsetmacro{\aTen}{90 - 360*10/24}
  \pgfmathsetmacro{\aSeventeen}{90 - 360*17/24}
  \coordinate (T2A) at ({\rtri*cos(\aZero)}, {\rtri*sin(\aZero)});
  \coordinate (T2B) at ({\rtri*cos(\aTen)}, {\rtri*sin(\aTen)});
  \coordinate (T2C) at ({\rtri*cos(\aSeventeen)}, {\rtri*sin(\aSeventeen)});
  \draw[ultra thick, blue] (T2A) -- (T2B) -- (T2C) -- cycle;

\end{tikzpicture}
}
   \caption{$I_0(C_n) = T_{10}(C_n) = F_n$.}
 \label{fig:inverting C neutral}
\end{figure}
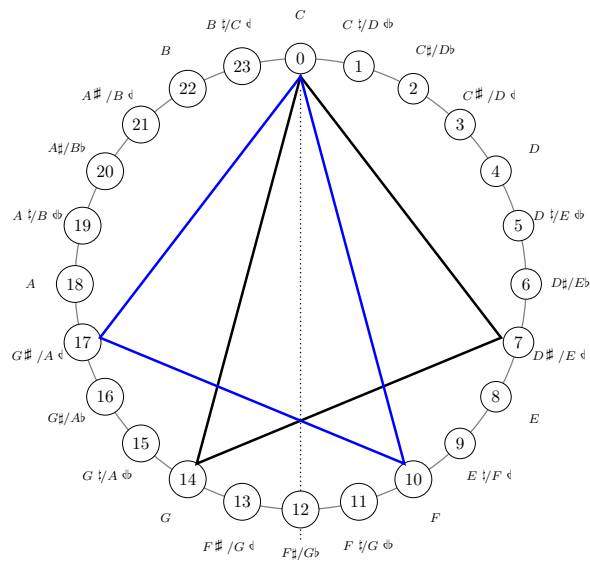

\section{The $PLR$-Groups}
So far, we have studied the action of the $TI$ group found in a quarter-tone system on the set $S$ of consonant triads and on the set $N$  of neutral triads.
In comparison to the original 12-tone system, the $TI$ group with 24 tones behaves as expected; in the 12-tone system, there exists a group isomorphic to $D_{12}$, and in the 24-tone system, there exists a group isomorphic to $D_{24}$. 
Moreover, in both cases, the $TI$ group acts on the set $S$ in a very similar manner. 
As shown in \cite{dihedralgroups}, there is a second dihedral group found in the consonant triads of the original 12-tone system.
This group is called the $PLR$-group, and is constructed using \textit{Parallel}, \textit{Leading Tone Exchange}, and \textit{Relative} functions, which will be further discussed in this section. 
One might expect that when we apply these functions to a quarter-tone system, we find another group isomorphic to $D_{24}$. Instead, there are two groups isomorphic to $D_{12}$. We will show this below with the proof of Theorem \ref{Thm:two groups}.

First, we define the $PLR$ functions. This definition is similar to the one in \cite{dihedralgroups}, except for the fact that the subindices are taken modulo 24. 
\begin{definition}\label{Def:PLR} The functions $P$ (Parallel), $L$ (Leading Tone Exchange), and $R$ (Relative) are functions from the set $S$ of all consonant (major and minor) triads to itself and are defined as follows:
    \begin{align*}
    P,L,R: S &\to S \notag \\
    P\langle y_1, y_2, y_3\rangle &= I_{y_1 + y_3}\langle y_1, y_2, y_3 \rangle, \\
    L\langle y_1, y_2, y_3\rangle &= I_{y_2 + y_3}\langle y_1, y_2, y_3 \rangle,\\
    R\langle y_1, y_2, y_3\rangle &= I_{y_1 + y_2}\langle y_1, y_2, y_3 \rangle.
\end{align*}
\end{definition}
Note that each of these functions is a bijection and can be represented by an inversion on the musical clock. However, unlike the Inversion function defined above, $P,L,$ and $R$ do not act component-wise; the axis of inversion is dependent on the input triad. 
\begin{example}
In this example we apply the $P$, $L$, and $R$ functions on the $C$-major triad:
\begin{align*}
P(C) = P\langle 0, 8, 14\rangle &= I_{14}\langle 0, 8, 14 \rangle \notag \\
&= \langle I_{14}(0), I_{14}(8), I_{14}(14) \rangle \notag \\
&= \langle 14, 6, 0 \rangle = c. \notag
\end{align*}
\begin{align*}
L(C) = L\langle 0, 8, 14\rangle &= I_{22}\langle 0, 8, 14 \rangle \notag \\
&= \langle I_{22}(0), I_{22}(8), I_{22}(14) \rangle \notag \\
&= \langle 22, 14, 8 \rangle = e. \notag
\end{align*}
\begin{align*}
R(C) = L\langle 0, 8, 14\rangle &= I_{8}\langle 0, 8, 14 \rangle \notag \\
&= \langle I_{8}(0), I_{8}(8), I_{8}(14) \rangle \notag \\
&= \langle 8, 0, 18 \rangle = a. \notag
\end{align*}
These three actions on $C$ are represented geometrically on the musical clock in Figures \ref{fig:Parallel on C}, \ref{fig:Leading tone on C}, and \ref{fig:Relative on C}.  
\end{example}

\begin{remark}
    Observing these figures, one might notice a couple of patterns. The first is that each of these functions, when applied to a major triad, yields a minor triad. When the input is a minor triad, the output is a major triad. Finally, when the input is a neutral triad, the output is also a neutral triad. 

A second pattern is that in each example, the two triangles on the musical clock share a side. This is not an accident. In music, it sounds interesting to move from one triad to another if they have notes in common, and the $P,L,R$ functions reflect this.
\end{remark}

More specifically, each of these functions fixes two of the components of a triad and moves the third component to create a new triad. For example, the Parallel function switches the first and third components. Above we saw that $P(C) = c$, which shares the first and third components. Further, $P(c) = C$.

This shows us that $P$ is its own inverse when applied to a consonant triad: because the parallel function fixes the first and third component, the axis of inversion does not change whether $C$ or $c$ is inputted. 
We can see this algebraically as well. 
Say $X \in S$ is a major triad. So $X = \langle x, x +8, x +14 \rangle $ for some pitch class $x$. Then we have
\begin{align}
    P\langle x, x+8, x+14\rangle &= I_{2x + 14}\langle x, x+8, x+14\rangle \notag \\
    &= \langle x + 2x + 14, -x-8+2x+14, -x-14+2x+14\rangle \notag \\
    &= \langle x+14, x+6, x\rangle, \notag 
\end{align}
and 
\begin{align}
    P\langle x + 14, x + 6, x \rangle &= I_{2x + 14} \langle x + 14, x + 6, x \rangle \notag \\
    &= \langle x, x + 8, x + 14\rangle. \notag
\end{align}
Hence, $P$ is its own inverse. A similar argument can be used to show $R$ and $L$ are their own inverses when applied to consonant triads.

When applied to neutral triads, the function $P$ acts as the identity, and the $L$ and $R$ functions can be written as translations.

The qualities we have just discussed of the $PLR$ functions on the set of consonant triads are consistent with the $PLR$ functions in the 12-tone system. However, before we formally define the $PLR$ groups, we first must discuss a quality that is unique to the 24-tone system.



\begin{figure}[htbp]
\centering
\scalebox{0.6}{
\begin{tikzpicture}
  \def\radius{5}
  \def\nodesize{0.6}
  \def\gapangle{8}
  \def\labeloffset{0.98}
  \def\innertriangleoffset{0.4}
  \pgfmathsetmacro{\rtri}{\radius - \innertriangleoffset}

  \newcommand{\notelabel}[1]{%
    \ifcase#1 $C$%
    \or $C\dsharp$/$D \sflat$%
    \or $C\sharp$/$D \flat$%
    \or $C\ssharp$/$D \dflat$%
    \or $D$%
    \or $D\dsharp$/$E \sflat$%
    \or $D\sharp$/$E \flat$%
    \or $D\ssharp$/$E \dflat$%
    \or $E$%
    \or $E\dsharp$/$F \dflat$%
    \or $F$%
    \or $F\dsharp$/$G \sflat$%
    \or $F\sharp$/$G \flat$%
    \or $F\ssharp$/$G \dflat$%
    \or $G$%
    \or $G\dsharp$/$A \sflat$%
    \or $G\sharp$/$A \flat$%
    \or $G\ssharp$/$A \dflat$%
    \or $A$%
    \or $A\dsharp$/$B \sflat$%
    \or $A\sharp$/$B \flat$%
    \or $A\ssharp$/$B \dflat$%
    \or $B$%
    \or $B\dsharp$/$C \dflat$%
    \fi
  }

  \foreach \i in {0,...,23} {
    \pgfmathsetmacro{\startangle}{90 - 360*\i/24 - \gapangle/2}
    \pgfmathsetmacro{\endangle}{90 - 360*(\i+1)/24 + \gapangle/2}
    \draw[thick, gray]
      ({\radius*cos(\startangle)}, {\radius*sin(\startangle)})
        arc[start angle=\startangle, end angle=\endangle, radius=\radius];
  }

  \pgfmathsetmacro{\aSeven}{90 - 360*7/24}
  \pgfmathsetmacro{\aNineteen}{90 - 360*19/24}
  \def\lineinner{4.7}  
  \def\lineouter{6.0}  

  \coordinate (LStart) at ({\lineinner*cos(\aSeven)}, {\lineinner*sin(\aSeven)});
  \coordinate (LEnd) at ({\lineouter*cos(\aNineteen)}, {\lineouter*sin(\aNineteen)});
  \draw[dotted, thick] (LStart) -- (LEnd);

  \foreach \i in {0,...,23} {
    \pgfmathsetmacro{\angle}{90 - 360*\i/24}
    \pgfmathsetmacro{\x}{\radius*cos(\angle)}
    \pgfmathsetmacro{\y}{\radius*sin(\angle)}

    \node[circle, draw, fill=white, minimum size=\nodesize cm]
      at (\x, \y) {\i};

    \pgfmathsetmacro{\lx}{(\radius + \labeloffset)*cos(\angle)}
    \pgfmathsetmacro{\ly}{(\radius + \labeloffset)*sin(\angle)}
    \node[font=\scriptsize, align=center] at (\lx, \ly) {\notelabel{\i}};
  }

  \pgfmathsetmacro{\aZero}{90 - 360*0/24}
  \pgfmathsetmacro{\aEight}{90 - 360*8/24}
  \pgfmathsetmacro{\aSix}{90 - 360*6/24}
  \pgfmathsetmacro{\aFourteen}{90 - 360*14/24}

  \coordinate (B0) at ({\rtri*cos(\aZero)}, {\rtri*sin(\aZero)});
  \coordinate (B8) at ({\rtri*cos(\aEight)}, {\rtri*sin(\aEight)});
  \coordinate (B14) at ({\rtri*cos(\aFourteen)}, {\rtri*sin(\aFourteen)});

  \pgfmathsetmacro{\aSix}{90 - 360*6/24}
  \coordinate (Bl0) at ({\rtri*cos(\aZero)}, {\rtri*sin(\aZero)});
  \coordinate (Bl6) at ({\rtri*cos(\aSix)}, {\rtri*sin(\aSix)});
  \coordinate (Bl14) at ({\rtri*cos(\aFourteen)}, {\rtri*sin(\aFourteen)});

  \draw[ultra thick, black] (B0) -- (B8) -- (B14) -- cycle;

  \draw[ultra thick, blue] (Bl0) -- (Bl6) -- (Bl14) -- cycle;

  \draw[line width=2pt, black] (B0) -- (B14);
  \draw[line width=1pt, blue] (B0) -- (B14);

\end{tikzpicture}
}
    \caption{$P(C) = c$.}
    \label{fig:Parallel on C}
\end{figure}


\begin{figure}[htbp]
\centering
\scalebox{0.6}{
\begin{tikzpicture}
  \def\radius{5}
  \def\nodesize{0.6}
  \def\gapangle{8}
  \def\labeloffset{0.98}
  \def\innertriangleoffset{0.4}
  \pgfmathsetmacro{\rtri}{\radius - \innertriangleoffset}

  \newcommand{\notelabel}[1]{%
    \ifcase#1 $C$%
    \or $C\dsharp$/$D \sflat$%
    \or $C\sharp$/$D \flat$%
    \or $C\ssharp$/$D \dflat$%
    \or $D$%
    \or $D\dsharp$/$E \sflat$%
    \or $D\sharp$/$E \flat$%
    \or $D\ssharp$/$E \dflat$%
    \or $E$%
    \or $E\dsharp$/$F \dflat$%
    \or $F$%
    \or $F\dsharp$/$G \sflat$%
    \or $F\sharp$/$G \flat$%
    \or $F\ssharp$/$G \dflat$%
    \or $G$%
    \or $G\dsharp$/$A \sflat$%
    \or $G\sharp$/$A \flat$%
    \or $G\ssharp$/$A \dflat$%
    \or $A$%
    \or $A\dsharp$/$B \sflat$%
    \or $A\sharp$/$B \flat$%
    \or $A\ssharp$/$B \dflat$%
    \or $B$%
    \or $B\dsharp$/$C \dflat$%
    \fi
  }

  \foreach \i in {0,...,23} {
    \pgfmathsetmacro{\startangle}{90 - 360*\i/24 - \gapangle/2}
    \pgfmathsetmacro{\endangle}{90 - 360*(\i+1)/24 + \gapangle/2}
    \draw[thick, gray]
      ({\radius*cos(\startangle)}, {\radius*sin(\startangle)})
        arc[start angle=\startangle, end angle=\endangle, radius=\radius];
  }

  \pgfmathsetmacro{\aEleven}{90 - 360*11/24}
  \pgfmathsetmacro{\aTwentyThree}{90 - 360*23/24}
  \def\lineinner{4.7}
  \def\lineouter{5.7}

  \coordinate (LStart) at ({\lineinner*cos(\aEleven)}, {\lineinner*sin(\aEleven)});
  \coordinate (LEnd) at ({\lineouter*cos(\aTwentyThree)}, {\lineouter*sin(\aTwentyThree)});
  \draw[dotted, thick] (LStart) -- (LEnd);

  \foreach \i in {0,...,23} {
    \pgfmathsetmacro{\angle}{90 - 360*\i/24}
    \pgfmathsetmacro{\x}{\radius*cos(\angle)}
    \pgfmathsetmacro{\y}{\radius*sin(\angle)}

    \node[circle, draw, fill=white, minimum size=\nodesize cm]
      at (\x, \y) {\i};

    \pgfmathsetmacro{\lx}{(\radius + \labeloffset)*cos(\angle)}
    \pgfmathsetmacro{\ly}{(\radius + \labeloffset)*sin(\angle)}
    \node[font=\scriptsize, align=center] at (\lx, \ly) {\notelabel{\i}};
  }

  \pgfmathsetmacro{\aZero}{90 - 360*0/24}
  \pgfmathsetmacro{\aEight}{90 - 360*8/24}
  \pgfmathsetmacro{\aFourteen}{90 - 360*14/24}
  \pgfmathsetmacro{\aTwentyTwo}{90 - 360*22/24}

  \coordinate (B0) at ({\rtri*cos(\aZero)}, {\rtri*sin(\aZero)});
  \coordinate (B8) at ({\rtri*cos(\aEight)}, {\rtri*sin(\aEight)});
  \coordinate (B14) at ({\rtri*cos(\aFourteen)}, {\rtri*sin(\aFourteen)});
  \draw[ultra thick, black] (B0) -- (B8) -- (B14) -- cycle;

  \coordinate (T8) at ({\rtri*cos(\aEight)}, {\rtri*sin(\aEight)});
  \coordinate (T14) at ({\rtri*cos(\aFourteen)}, {\rtri*sin(\aFourteen)});
  \coordinate (T22) at ({\rtri*cos(\aTwentyTwo)}, {\rtri*sin(\aTwentyTwo)});
  \draw[ultra thick, blue] (T8) -- (T14) -- (T22) -- cycle;

  \draw[line width=2pt, black] (B8) -- (B14);
  \draw[line width=1pt, blue] (B8) -- (B14);

\end{tikzpicture}
}
    \caption{$L(C) = e$.}
    \label{fig:Leading tone on C}
\end{figure}

\begin{figure}[htbp]
\centering
\scalebox{0.6}{
\begin{tikzpicture}
  \def\radius{5}
  \def\nodesize{0.6}
  \def\gapangle{8}
  \def\labeloffset{0.98}
  \def\innertriangleoffset{0.4}
  \pgfmathsetmacro{\rtri}{\radius - \innertriangleoffset}

  \newcommand{\notelabel}[1]{%
    \ifcase#1 $C$%
    \or $C\dsharp$/$D \sflat$%
    \or $C\sharp$/$D \flat$%
    \or $C\ssharp$/$D \dflat$%
    \or $D$%
    \or $D\dsharp$/$E \sflat$%
    \or $D\sharp$/$E \flat$%
    \or $D\ssharp$/$E \dflat$%
    \or $E$%
    \or $E\dsharp$/$F \dflat$%
    \or $F$%
    \or $F\dsharp$/$G \sflat$%
    \or $F\sharp$/$G \flat$%
    \or $F\ssharp$/$G \dflat$%
    \or $G$%
    \or $G\dsharp$/$A \sflat$%
    \or $G\sharp$/$A \flat$%
    \or $G\ssharp$/$A \dflat$%
    \or $A$%
    \or $A\dsharp$/$B \sflat$%
    \or $A\sharp$/$B \flat$%
    \or $A\ssharp$/$B \dflat$%
    \or $B$%
    \or $B\dsharp$/$C \dflat$%
    \fi
  }

  \foreach \i in {0,...,23} {
    \pgfmathsetmacro{\startangle}{90 - 360*\i/24 - \gapangle/2}
    \pgfmathsetmacro{\endangle}{90 - 360*(\i+1)/24 + \gapangle/2}
    \draw[thick, gray]
      ({\radius*cos(\startangle)}, {\radius*sin(\startangle)})
        arc[start angle=\startangle, end angle=\endangle, radius=\radius];
  }

  \pgfmathsetmacro{\aFour}{90 - 360*4/24}
  \pgfmathsetmacro{\aSixteen}{90 - 360*16/24}
  \def\lineinner{4.7}
  \def\lineouter{5.6}
  \coordinate (LStart) at ({\lineinner*cos(\aFour)}, {\lineinner*sin(\aFour)});
  \coordinate (LEnd) at ({\lineouter*cos(\aSixteen)}, {\lineouter*sin(\aSixteen)});
  \draw[dotted, thick] (LStart) -- (LEnd);

  \foreach \i in {0,...,23} {
    \pgfmathsetmacro{\angle}{90 - 360*\i/24}
    \pgfmathsetmacro{\x}{\radius*cos(\angle)}
    \pgfmathsetmacro{\y}{\radius*sin(\angle)}

    \node[circle, draw, fill=white, minimum size=\nodesize cm]
      at (\x, \y) {\i};

    \pgfmathsetmacro{\lx}{(\radius + \labeloffset)*cos(\angle)}
    \pgfmathsetmacro{\ly}{(\radius + \labeloffset)*sin(\angle)}
    \node[font=\scriptsize, align=center] at (\lx, \ly) {\notelabel{\i}};
  }

  \pgfmathsetmacro{\aZero}{90 - 360*0/24}
  \pgfmathsetmacro{\aEight}{90 - 360*8/24}
  \pgfmathsetmacro{\aFourteen}{90 - 360*14/24}
  \pgfmathsetmacro{\aEighteen}{90 - 360*18/24}

  \coordinate (N0) at ({\rtri*cos(\aZero)}, {\rtri*sin(\aZero)});
  \coordinate (N8) at ({\rtri*cos(\aEight)}, {\rtri*sin(\aEight)});
  \coordinate (N14) at ({\rtri*cos(\aFourteen)}, {\rtri*sin(\aFourteen)});
  \coordinate (N18) at ({\rtri*cos(\aEighteen)}, {\rtri*sin(\aEighteen)});

  \draw[ultra thick, black] (N0) -- (N8) -- (N14) -- cycle;

  \draw[ultra thick, blue] (N0) -- (N8) -- (N18) -- cycle;

  \draw[line width=2pt, black] (N0) -- (N8);
  \draw[line width=1pt, blue] (N0) -- (N8);

\end{tikzpicture}
}
   \caption{$R(C) = a$.}
    \label{fig:Relative on C}
\end{figure}

\begin{lemma}\label{Lemma:original tone under PLR}
    Let $X \in S$ be triad with root $x$ where $x$ is an original tone. Then $P(X), L(X),$ and $R(X)$ are triads with original tones as roots. 
\end{lemma}

\begin{proof}
First, observe on the musical clock that each original tone is attributed with an even number. 
Let $X \in S$ be a triad with root $x$, which is an original tone. 
Then $X$ is either major or minor. 
Assume $X$ is major. Then $X$ takes the form $\langle x, x + 8, x + 14 \rangle$. 
As shown above, $P \langle x, x + 8, x + 14 \rangle = \langle x + 14, x + 6, x \rangle$. Since $x$ is even, the minor triad $\langle x + 14, x + 6, x \rangle$ also has an original tone as a root.  
Now assume $X \in S$ is minor. Then $X$ takes the form $\langle x, x + 6, x +14 \rangle$ where $x$ is even. 
Then $P\langle x, x + 6, x + 14 \rangle = \langle x + 14, x + 8, x  \rangle$, which is a major triad with an original tone root. 
A similar argument can be used to show that $L(X)$ and $R(X)$ yield triads with original tones as roots.
\end{proof}

\begin{lemma}
Let $Y \in S$ be a triad with root $y$ where $y$ is a new tone. Then $P(Y), L(Y)$, and $R(Y)$ are triads with new tones as roots. 
\end{lemma}

\begin{proof}
    Observe that each of the new tones on the musical clock corresponds to an odd number. 
    Let $Y \in S$ be a triad with root $y$ where $y$ is a new tone. Then $Y$ is either major or minor. 
    Assume $Y$ is major. Then, similarly to above, $P(Y) = \langle y + 14, y + 6 , y\rangle$, which is a minor triad with a new tone as a root.
    Now, assume $Y$ is minor. Then $P(Y) = \langle y + 14, y +8, y \rangle$, which is a major triad with a new tone as a root. 
    A similar argument can be used to show $L(Y)$ and $R(Y)$ are triads with new tones as roots. 
\end{proof}
\begin{lemma}
    Let $Z \in N$ be a neutral triad with root $x$. The function $P(Z)$ fixes the triad, and both $R(Z)$ and $L(Z)$ switch the nature of the root tone from original tone to new tone and vice versa.
    \begin{proof}
        This follows from a direct application of Definitions \ref{Def:PLR} and \ref{Def:transposition and inversion formulas}.
    \end{proof}
\end{lemma}
We can now formally define the $PLR$ groups.    
\begin{theorem} \label{Thm: cyclic}
   The $L$ function on the set of neutral triads $N$ generates the cyclic group $\mathbb{Z}_{24}$.
\end{theorem}
\begin{proof}
    As noted above, the $L$ and $R$ functions can be defined as translations. In particular, 
    \begin{align*}
        L\langle x, x+7, x+14 \rangle &= I_{2x+21} \langle x, x+7, x+14 \rangle \\
                                       &= \langle x+21, x+14, x+7 \rangle \\
                                       & = \langle x+7, x+14, x+21 \rangle \\
                                       &= T_7 \langle x, x+7, x+14 \rangle.
    \end{align*}
    A similar computation shows that 
    \begin{align*}
        R\langle x, x+7, x+14 \rangle &= \langle x-7, x, x+7\rangle \\
                                       &= T_{17}\langle x, x+7, x+14 \rangle.
    \end{align*}
    Also note that $R = (L)^{23}$.
    Since 7 and 24 are coprime, $L$ generates a cyclic group of order 24. That is, we can apply $L$ 24 times and obtain all neutral triads. The sequence of root notes in this case is 
    $$
    C, D \ssharp, G, A \ssharp, D, F \dsharp, A, C \dsharp, E, G \dsharp, B, D \dsharp, $$
    $$F \sharp, A \dsharp, C \sharp, E \dsharp, G \sharp, B \dsharp, D \sharp, F \ssharp, A \sharp, C \ssharp, F, G \ssharp.
    $$
  Enharmonic pitches can also be used to express this same sequence in different ways.   
\end{proof}

\begin{theorem}\label{Thm:two groups}
The $L$ and $R$ functions on the set of consonant triads $S$ generate two groups isomorphic to $D_{12}$.
\end{theorem}

\begin{proof}
    First, consider the set of all consonant triads with original tone roots. This is a subset of $S$ and has order 24. 
    If we begin with a $C$ major triad and alternatively apply $R$ and $L$, the result is the following sequence of triads:
    $$
    C, a, F, d,B\f, g, E\f, c, A\f, f, D\f, b\f, G\f, e\f, B, g\s, E, c\s, A, f\s, D, b, G, e, C.
    $$
    This sequence does not repeat itself until it returns to $C$, it is 24 triads long, and each original tone triad appears. Thus the parallel of $C$ is present in this sequence; meaning we can write $P$ as a composition of $L$ and $R$. More specifically, $P = R(LR)^3$.
    Further, we  know that each of the bijections $R, LR, RLR, ... R(LR)^{11}$ are distinct, and that $(LR)^{12} = 1$. 
 
    If we set $x = LR$ and $y = L$ then $y^2 = 1$, $x^{12} = 1$, and 
    \begin{align}
        yxy &= L(LR)L \notag \\
        &= RL \notag \\
        &= x^{-1}. \notag 
    \end{align}
    This subset of even-numbered consonant triads has precisely 24 elements. Moreover, by Lemma \ref{Lemma:original tone under PLR}, we know any combination of $L$ and $R$ will yield an original tone triad. Thus, by the definition in \cite{algebra}, the original toned triads and the $LR$ functions generate a group isomorphic to $D_{12}$. 

    Now consider the set of all consonant triads with new tone roots (triads where $x$ is odd). This is also a subset of $S$ with order 24. 
    We can similarly apply $R$ and $L$ to $C$~{\hskip-6pt\dsharp}to yield the following sequence of triads:
    $$
    C \dsharp, a\dsharp, F \dsharp, d \dsharp, B\dflat, g\dsharp, E\dflat, c \dsharp, A\dflat, f\dsharp, D\dflat, b\dflat, $$
    $$
    G\dflat, e\dflat, B\sharp, g \ssharp, E\dsharp, c\ssharp, A\dsharp, f\ssharp, D\dsharp, b\dsharp, G\dsharp, e\dsharp, C\dsharp.
    $$
Thus with $s = LR$ and $t = L$ we have the relations $x^{12} = 1$, $t^2 = 1$, and $tst = s^{-1}$. Using a similar argument to above, the set of new tone triads with the $LR$ functions generate a group isomorphic to $D_{12}$.
\end{proof}

One may notice that the two sequences of triads given in the proof are similar; in fact, they are the same sequence except the second is shifted up by a quarter step. We mentioned previously that we can think of the 24-tone system as 12-tone systems, one shifted up by a quarter step. The existence of these groups isomorphic to $D_{12}$ is an algebraic reflection of this concept. Rather than finding one group isomorphic to $D_{24}$, we have found two groups isomorphic to $D_{12}$ which reflect the relationship between the original and new tones. 

\begin{example}
The algebraic result from Theorem \ref{Thm:two groups} was observed from the musical point of view by Wyschnegradsky in \cite{manual}, where he describes the
24-tone system as two 12-tone systems interlocked.  This observation is realized very clearly in the excerpt in Figure \ref{fig:two_piano}.

\begin{figure}[htbp]
      \centering
      \includegraphics[width=0.85\linewidth]{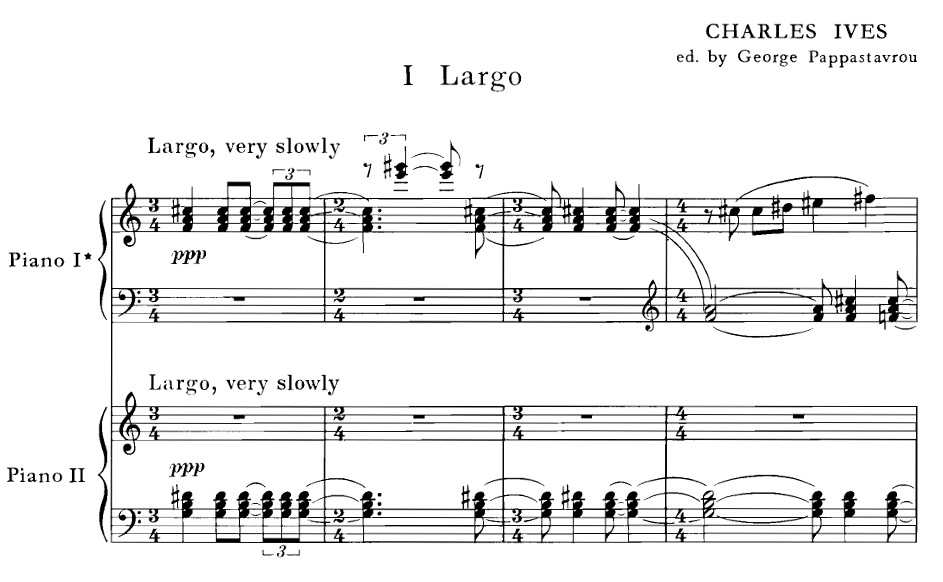}
    \caption{ 
Excerpt from: Three Quarter-tone Pieces by Charles Ives.
\copyright\ 1968 by C.F. Peters Corporation.
Used by permission of Faber Music Ltd. All Rights Reserved. [EP66285].
\textit{*Piano I to be tuned $\frac{1}{4}$ tone higher.}
}
      \label{fig:two_piano}
  \end{figure}

\normalfont
The music score in Figure \ref{fig:two_piano} is an excerpt from a composition for piano duet where one piano (Piano I) is tuned a quarter-step above the other piano (Piano II). So, Piano I plays exclusively new tones while piano II plays original tones, realizing not only Wyschnegradsky's musical observation but also the mathematical underlying structure. 
\end{example}

\section*{Acknowledgments}
The first-named author would like to thank her mentor, Dr. Carmen Rovi, for her help, guidance, and coffee throughout this process. Both authors are grateful to Dr. Emily Peters for initially recommending \cite{dihedralgroups}, the paper that this study extends and to Dr. Dongryul Lee for sharing literature on microtonal music. Special thanks also go to Cecily Bartsch and Grace Houghton for their thoughtful comments during the development of this work. The first-named author also thanks her mother for her endless love and support.
The second-named author gratefully acknowledges support from the College of Arts and Sciences at Loyola University Chicago for summer research funding, as well as the AMS and Simons Foundation for support through an AMS-Simons Research Enhancement Grant for Primarily Undergraduate Institution Faculty. Both authors are grateful to Farber Music Ltd. for allowing them to reproduce the music excerpt by Charles Ives.

\bibliographystyle{plain}
\bibliography{refs}

\begin{thebibliography}{1}

\bibitem{algebra}
Michael Artin.
\newblock {\em Algebra Second Edition}.
\newblock Pearson Modern Classic, 2018.

\bibitem{dihedralgroups}
Alissa~S. Crans, Thomas~M. Fiore, and Ramon Satyendra.
\newblock Musical actions of dihedral groups.
\newblock {\em Amer. Math. Monthly}, 116(6):479--495, 2009.

\bibitem{fiftysystem}
Ben Johnston.
\newblock Scalar order as a compositional resource.
\newblock {\em Perspectives of New Music}, 2(2):56--76, 1964.

\bibitem{manual}
I.~Wyschnegradsky, N.~Kaplan, and R.~Kaplan.
\newblock {\em Manual of Quarter-tone Harmony}.
\newblock Underwolf Editions, 2017.

\end{thebibliography}

\medskip

\medskip

\noindent Veronica Flynn: Department of Mathematics \& Statistics, Loyola University Chicago. \texttt{vflynn@luc.edu}

\noindent Carmen Rovi: Department of Mathematics \& Statistics, Loyola University Chicago.  \\ \texttt{crovi@luc.edu}


\end{document}